%% file: main.tex
\documentclass[sigconf]{acmart}

\usepackage{subfigure}
\usepackage{color}
\usepackage{multirow}
\usepackage{placeins}
\usepackage{float}

\AtBeginDocument{%
  }

\setcopyright{none}

\acmConference[LADC '22]{11th Latin-American Symposium on Dependable Computing}{November 21--24,
  2022}{Fortaleza, BR}

\newtheorem{theorem}{Theorem}[section]
\newtheorem{corollary}{Corollary}[section]

\begin{document}
\pagestyle{plain}

\title{A Distributed System-level Diagnosis Model for the Implementation of Unreliable Failure Detectors}

\author{Elias P. Duarte Jr., Luiz A. Rodrigues, Edson T. Camargo, and Rogerio Turchetti}
\email{elias@inf.ufpr.br, luiz.rodrigues@unioeste.br, edson@utfpr.edu.br, turchetti@redes.ufsm.br}
\affiliation{
  \institution{Federal University of Parana (UFPR) Western Parana State  University (UNIOESTE)\\ Technological Federal University of Parana (UTFPR) Federal University of Santa Maria (UFSM)}
  \city{Curitiba (PR), Cascavel (PR), Toledo (PR), Santa Maria (RS)}
  \country{BRAZIL}
}

\renewcommand{\shortauthors}{Duarte Jr. et al.}

\input{Sections/0-abstract.tex}


\maketitle
\raggedbottom


\input{Sections/1-intro-drct-ladc.tex}
\input{Sections/2-sld-drct-ladc.tex}
\input{Sections/3-ufd-drct-ladc.tex}

\input{Sections/4-model-drct-ladc.tex}
\input{Sections/5-props-drct-ladc.tex}

\input{Sections/6-latency-drct-ladc}

\input{Sections/7-threefd-drct-ladc.tex}
\input{Sections/8-conc-drct-ladc.tex}

\bibliographystyle{ACM-Reference-Format}
\bibliography{references}

\end{document}

%% file: Sections/0-abstract.tex
\begin{abstract}

Reliable systems require effective monitoring techniques for fault identification. System-level diagnosis was originally proposed in the 1960s as a test-based approach to monitor and identify faulty components of a general system. Over the last decades, several diagnosis models and strategies have been proposed, based on different fault models, and applied to the most diverse types of computer systems. In the 1990s, unreliable failure detectors emerged as an abstraction to enable consensus in asynchronous systems subject to crash faults. Since then, failure detectors have become the \textit{de facto} standard for monitoring distributed systems. The purpose of the present work is to fill a conceptual gap by presenting a distributed diagnosis model that is consistent with unreliable failure detectors. Results are presented for the number of tests/monitoring messages required, latency for event detection, as well as completeness and accuracy. Three different failure detectors compliant with the proposed model are presented, including vRing and vCube which provide scalable alternatives to the traditional all-monitor-all strategy adopted by most existing failure detectors.

\end{abstract}

%% file: Sections/1-intro-drct-ladc.tex
\section{Introduction}\label{sec1}

As computer systems have become part of the very fabric of human organizations, failures can have serious consequences \cite{NYT:2021,Codestone:2017}. It is essential to project systems that continue to work correctly even if some components become faulty. Fault tolerance emerged soon after the development of the first digital computers. In the 1950s, Von Neumann himself investigated the construction of reliable computers based on unreliable components \cite{vonNeumann:1956}. Today dependability is a well-structured domain that encompasses the properties that reflect the degree of confidence that can be placed in a system \cite{dependability:2004}. 

There are several techniques for building fault-tolerant systems \cite{GoogleSRE:2016, pradhan:1996}. Most of those techniques exploit redundancy, both explicitly and implicitly. In explicit redundancy, components are replicated to avoid single points of failure. For example, instead of a single secondary memory unit, multiple units are employed so that the failure of a single unit does not affect the correct operation of the system as a whole. Redundancy can also be implicit, being part of the system, as in fault-tolerant routing \cite{duarte2004delivering}. In the case of distributed systems, redundancy is implicit and intrinsic, since a distributed system by definition consists of a set of $n \geq 2$ processes that communicate and collaborate to accomplish a task \cite{reynal2005short}. 
Among the dependability properties, availability indicates the percentage of time a system is expected operate correctly, taking into account failures and recoveries. In order to increase availability, it is essential to have an effective failure recovery mechanism. The idea is to reduce the time the system remains unavailable as much as possible. Recovery usually begins with the determination that a failure has occurred. Although some dependable systems do not require explicit fault identification -- for example, those that make decisions based on majority voting -- most fault-tolerant systems follow the classical model of fault identification, fault isolation, and system reconfiguration \cite{pradhan:1996}. 

Fault identification is also an old problem: the first model of diagnosable systems appeared in the 1960s: the PMC model, named after the authors' initials \cite{pmc:1967}.  According to the PMC model, system units perform tests on each other. The test results can be used to determine which units are faulty. The PMC model assumes that a correct tester can accurately determine and report the actual state of each tested unit. Based on the PMC model, an extraordinary number of results have been obtained in the field, including various models applied to different types of systems. Those results have allowed the understanding of the limitations of diagnosis and the to define a wide variety of strategies for fault identification \cite{blough:1996,duarte:2011}.

In the context of distributed systems, the so-called FLP impossibility was published in 1985, also named after the authors' initials \cite{flp:1985}. According to this key result, it is impossible to guarantee the correct execution of consensus in asynchronous distributed systems, where processes may crash. Asynchronous systems have no timing guarantees: there are no known bounds on the maximum time required to execute a task and transfer messages between processes. The root of FLP impossibility is precisely the difficulty of distinguishing a faulty process from a slow process. Given the importance of consensus -- considered by many to be the central problem of distributed systems -- this result has had very significant implications for the field as a whole.

In the 1990s, Chandra and Toueg \cite{failure-detectors:1996} looked at the consensus problem from a different angle. The main idea was to investigate how the FLP impossibility would be affected if processes had information about failures. The authors defined a failure detector that acts as an oracle and provides information about process states. A process is reported as either correct or suspected of having failed. Failure detectors are inherently unreliable, meaning that the detector may report a process state that does not correspond to reality. The authors have defined two failure detection properties: completeness and accuracy. Informally, completeness reflects the ability of the detector to identify processes that have actually failed. Accuracy, on the other hand, is the ability of the detector to not suspect that correct processes have failed.

Most failure detectors use a monitoring strategy that consists of each monitored processes send sending \textit{heartbeat} messages to all processes at regular intervals \cite{sens:2002,qosconfigurable:2016}. If the system runs on a single network segment, it is possible to implement this strategy efficiently, for example using multicast in hardware. In other environments, the strategy does not scale well, because it requires periodic transmission of $n^2$ messages. Most efforts to develop scalable detectors have involved probabilistic \cite{gupta2001scalable} broadcasting. On the other hand, going back to distributed diagnosis algorithms, they have been proposed precisely to reduce the number of messages required for system monitoring, as well as the latency for detecting new process state changes. 

This work brings an effort to unify distributed diagnosis and failure detection. A new system-level diagnosis model is proposed that enables the specification of scalable failure detectors.  Results are presented on bounds on the number of tests/monitoring messages required, latency for event detection, and completeness and accuracy.Three different failure detectors are presented that are consistent with the proposed model. The first one is based on the traditional all-monitor-all strategy adopted by most existing failure detectors. Then, two classical diagnosis algorithms were used to the new model: vRing and vCube, which provide scalable alternatives to the traditional all-monitor-all strategy used by most existing failure detectors.

The remainder of this paper is organized as follows. The next section provides an overview of system-level diagnosis. Section 3 defines and gives an overview of failure detectors. Section 4 introduces the distributed diagnosis/failure detector model and presents basic results, including the number of tests required.  Section 5 shows results for completeness and accuracy of fault diagnosis/detection. Section 6 presents proofs on the best and worst detection latency for the model. Section 7 presents the three failure detectors for the proposed model. Finally, the conclusion can be found in Section 8.

%% file: Sections/2-sld-drct-ladc.tex
\section{System-Level Diagnosis: An Overview}\label{sec2}

In 1967, Preparata, Metze, and Chien \cite{pmc:1967} published the first model of ``diagnosable'' systems, called the PMC model after the authors' initials. Previous approaches to identify faults in computer systems had been at the level of individual components. In the PMC model, the system is composed of units capable of testing each other. A test consists of a procedure complete enough to determine whether the tested unit is \textit{faulty} or \textit{fault-free}. The set of all test results is called the \textit{syndrome} of the system. A central entity external to the system collects and processes the syndrome to classify units as faulty or fault-free.

Interestingly, the PMC model adopted a fault model that corresponds to Byzantine errors today. Although no threats or malicious components are assumed, a faulty unit produces arbitrary outputs. More specifically, faulty units perform tests and report test results that can have arbitrary value. On the other hand, the PMC model assumes that a fault-free unit runs the tests correctly and reports the correct test results. Thus, depending on the amount of tests performed and the state of the testers, the syndrome may or may not allow correct identification of faulty units. Figure~\ref{fig:PMC} shows a classic PMC example. The system consists of $n=5$ units connected by arcs representing tests with labels indicating the results: 1 for faulty and 0 for fault-free. In the example, at most a single unit can be faulty to correctly determine the states of all units from processing the syndrome. If, on the other hand, there are two or more faulty units, the problem becomes impossible: there is no way to tell which unit is faulty. The concept of \textit{diagnosability} was defined to reflect the ability of a system to diagnose $f$ failures. A $f$-diagnosable system can correctly identify up to $f$ faulty units. The system in Figure~\ref{fig:PMC} is 1-diagnosable.

\begin{figure}[htb]
\centering
\includegraphics[scale=0.65]{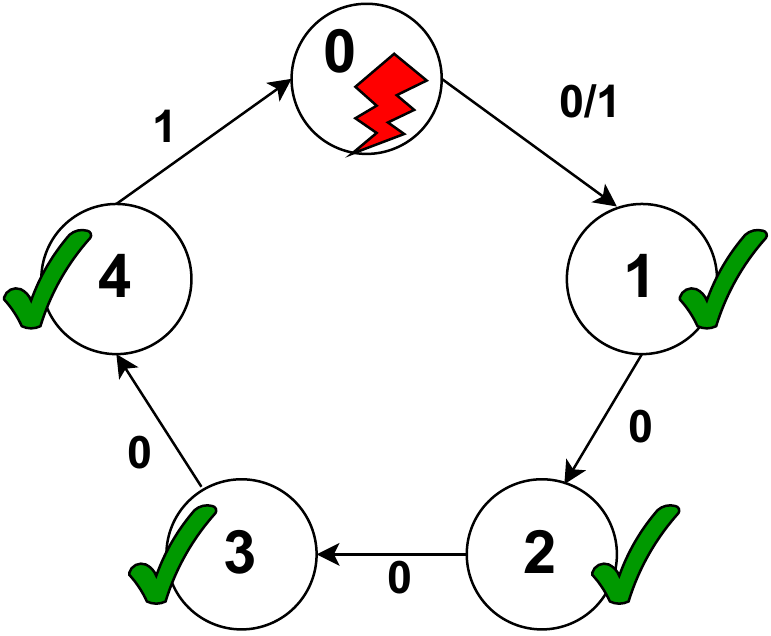}
\caption{The classic PMC model example.}
\label{fig:PMC}
\end{figure}

The set of tests performed on the system was originally called \textit{connection assignment} and later came to be called \textit{testing assignment}. In the early years of diagnosis, much of the research focused on determining testing assignments with a low number of tests and high \textit{diagnosability}, which would allow effective and efficient diagnosis. In 1974, Hakimi and Amim proved that for a system of $n$ units to be $t$-diagnosable it is necessary that (i) $n \geq 2t+1$ and that (ii) each process be tested for at least $t$ other processes \cite{hakimi1974characterization}.

Ten years later, in 1984, two new diagnosis models were proposed that had a major impact on the field. The first model by Hakimi and Nakajima allows diagnosis to be adaptive, in the sense that the testing assignment is dynamic i.e., the next tests to be executed are defined according to the results of the previous tests \cite{hakimi1984adaptive}. The possibility of employing adaptive tests introduces a time dimension to diagnosis: at first, a set of tests are executed, then their results are evaluated, and finally new tests are defined for execution in the next round. The other model was proposed by Hosseini, Kuhl, and Reddy \cite{hosseini1984diagnosis}, and eliminated the central entity, thus allowing fully distributed diagnosis.  According to that model, the units not only execute tests but also collect test results from each other to obtain the syndrome. Each fault-free unit can then process the syndrome and determine the set of faulty units in the system.

It was only in the following decade, in 1992, that Bianchini and Buskens proposed the combination of adaptive and distributed diagnosis \cite{bianchini1992implementation}. The Adaptive-DSD (\textit{Adaptive Distributed System-Level Diagnosis}) algorithm was proposed and used to monitor a large production network. The algorithm is based on a ring topology, presenting the optimal number of tests: at most $n$, but the latency is of up to $n$ consecutive testing rounds. In 1998, Duarte and Nanya proposed the Hi-ADSD (Hierarchical Adaptive Distributed System-Level Diagnosis) algorithm that is also based on an adaptive and distributed diagnosis model  \cite{duarte1998hierarchical}. Hi-ADSD employs a hierarchical virtual topology to organize the system units. The latency is at most of $log^2n$ testing rounds, all logarithms in this work are base 2.

The Hi-ADSD algorithm was originally proposed for the construction of efficient and fault-tolerant network fault management systems, having been implemented using the Internet management protocol SNMP (\textit{Simple Network Management Protocol }) \cite{debona2004flexible}. Later, the algorithm was slightly modified to ensure that the number of tests executed does not exceed $nlogn$ per $n$ testing rounds \cite{duarte2014vcube}. In the new version, the algorithm is called VCube, and has been used as a failure detector to build reliable distributed algorithms \cite{rodrigues2014autonomic}.  

The system-level diagnosis results presented above assume the PMC model, and a fully connected system, i.e., the topology is a complete graph. There are actually several other models, for instance, those that assume a system of arbitrary topology \cite{duarte2011distributed}, which are partitionable by definition. In those models, the problem of determining which units are faulty becomes the problem of computing reachability. Other works consider different diagnosis models. Camargo and others \cite{camargo2018running} have dropped the PMC model assumption to assume correct units can make mistakes as they execute a test. Comparison-based diagnostics \cite{duarte:2011,ziwich2016nearly} does not assume crash faults. Faulty units can produce arbitrary outputs. The outputs of multiple units are compared to perform diagnosis. An example application of comparison-based diagnosis is distributed integrity checking \cite{ziwich2005distributed}. Finally, the probabilistic diagnosis model also presents an alternative approach, which takes into account the difficulty of defining deterministic tests \cite{blough:1996}.

%% file: Sections/3-ufd-drct-ladc.tex
\section{An Overview of Unreliable Failure Detectors}\label{sec3}

The output produced by an unreliable fault detector is exactly the same as that of a distributed system-level diagnosis algorithm: the list of processes considered to be faulty. Despite this, failure detectors were proposed in a completely different context \cite{failure-detectors:1996, delporte2004weakest}. The original motivation for its development was the impossibility of consensus in asynchronous systems with crash faults. Chandra and Toueg investigated what would happen if a set of processes executing consensus had access to process state information. In a way, they investigated the very limits of the FLP impossibility. The question here is whether this ``extra'' information would allow consensus to be possible.

Thus, failure detectors were not proposed to be efficient strategies for process monitoring. The focus was on the properties a failure detector should have to help solve consensus.  Two essential properties were then defined: \textit{completeness} and \textit{accuracy}.  Informally, completeness reflects the ability of a failure detector to suspect failed processes. Accuracy, on the other hand, requires the failure detector not to suspect correct processes. Regarding completeness, two observations can be made. First, as crash faults are assumed, a process that has failed does not produce any response to any stimulus, so it is not difficult to guarantee that the failure detector raises suspicions of what has occurred. On the other hand, the monitored process may fail immediately after the failure detector had successfully exchanged messages that were used to determine it was correct. Thus it may take some time after the failure occurred for the suspicion to be raised.

Accuracy, on the other hand, can never be guaranteed in asynchronous systems. The problem is classic: it is difficult (impossible, in fact) to distinguish a process that has failed from a slow process. A correct process can be slow to execute tasks or communicate. That is exactly the root of the FLP impossibility. However, when the failure detector does not suspect that correct processes have failed, what impact can it have on consensus? In \cite{chandra1996weakest}, the authors show that even if eventually a single correct process is never suspected,  then consensus can be solved in an asynchronous distributed system with \textit{crash} faults. That is a far-reaching result, with consequences both in terms of theory and practice. That unsuspected correct process can be elected as the leader and have multiple responsibilities in a distributed algorithm.

Completeness and accuracy can be classified as either weak or strong, resulting in a total of eight classes of failure detectors, as described next. Completeness is strong if, after a time interval, all failed processes are suspected by all correct processes. On the other hand, completeness is weak if, after a time interval, all failed processes are suspected by at least one correct process. It is not difficult to convert weak completeness into strong completeness: the correct process that detected the failure can broadcast the information to the other correct processes.

Accuracy can also be strong or weak. Strong accuracy means that the failure detector does not suspect that any correct processes have failed. Weak accuracy only requires that at least a single correct process is never suspected. These two properties are further extended as eventual weak/strong accuracy: after a time interval no/one (respectively) correct process is suspected of having failed. The weakest class of failure detectors (called $\diamond W$) presents weak completeness and eventual weak accuracy. Even that class of failure detector guarantees the correct execution of consensus in asynchronous systems with crash faults, as mentioned above. Note that an asynchronous system augmented with a failure detection is no longer ``purely'' asynchronous, as it involves timing properties.

Classic implementations of failure detectors require the periodic transmission of heartbeat messages by each monitored process to all others \cite{sens:2002, chen2002quality, moraes2011failure}. This strategy is efficient if the distributed system is running on a single physical network based on broadcast, as it takes a single message to send each heartbeat to all processes. On the other hand, for systems running on point-to-point networks, or on more than one network, the strategy does not scale, requiring the periodic transmission of $n^2$ messages. That is one of the motivations for defining the system-level diagnosis model for building failure detectors based on tests presented in the next section.

%% file: Sections/4-model-drct-ladc.tex
\section{A System-Level Diagnosis Model for Unreliable Failure Detection}

In this section we present the new distributed diagnosis model for the implementation of classic failure detectors. The model assumes a distributed system $\Pi$ that consists of $n$ processes which communicate with message passing. Processes have sequential identifiers, thus $\Pi = \{p_0, p_1, ..., p_{n-1}\}$. Processes can be also identified by their ids alone, $\Pi = \{0, 1, ..., n-1\}$. The system is fully connected, and is represented by an undirected complete graph $K_n = (\Pi,E)$, where $E$ is the set of edges, $E=\{\{i,j\} \, | \, 0 \leq i, j < n$ e $i \neq j\}$. Given this model, any two processes can communicate directly without having to pass through intermediaries.

The system is asynchronous, with no known limits on the maximum time required to transmit a message or to execute a task. The crash fault model is assumed, thus a failed process looses its internal state completely and does not produce any output for any input. A process can be in one of two states: \textit{failed} ou \textit{correct}. A \textit{correct/failed} process is also said to be \textit{fault-free}/\textit{crashed}, respectively. Function \textit{state(i)} returns the state of process $i$: $\mbox{\textit{state}}(i) =$ \{\textit{failed} $\mid$ \textit{correct}\}. An \textit{event} is defined as the transition of the state of a process, from \textit{correct} to \textit{failed}. 
Although it is straightforward to extend the model to include process recovery, as well to allow a process to suffer multiple events along the time, in this work, for the sake of space, only permanent crash faults with no recovery are assumed.

Processes test each other. The purpose of a test is to allow the tester to determine the state of the tested process. The outcome of a test can indicate either that the tested process is  \textit{correct} ou \textit{suspect} of having failed. A test consist of a set of stimuli send by tester to the tested process, for which the proper replies are expected. The specific testing procedure adopted varies according to the technology of the system, as well as other factors such as specific functionalities can be checked to be correct. After the tester receives the proper reply, it classifies the tested process as correct. Consider two processes $i,j \in \Pi$, a test executed by tester $j$ on tested process $i$ is defined as function $\mbox{\textit{test}}_j(i)$ = \{\textit{correct} $\mid$ \textit{suspect}\}. A proper timeout mechanism has to be adopted to limit the time the tester will wait for a reply \cite{qosconfigurable:2016}. The testing procedure has to be strong enough to avoid mistakes as much as possible, thus for instance the decision to classify the tested process as suspect must be taken after more than a single timeout. Corollary \ref{col1} establishes a correlation between the states of a tested process and its classification after a test.

\begin{corollary}
\label{col1}
A process that is tested as \textit{correct} is indeed in this state, thus $\forall i, j \in V$, if $\mbox{\textit{test}}_j(i)$ = \textit{correct}, then \textit{state}(i)= \textit{correct}. On the other hand, a process classified as \textit{suspect} can be either \textit{failed} or \textit{correct}, in this case due to slowness that caused a timeout as it was tested. In this way, if $\mbox{\textit{test}}_j(i)$ = \textit{suspect}, state \textit{state}(i)=\{\textit{correct} $\mid$ \textit{failed}\}.
\end{corollary}

The set of tests is called \textit{testing assignment}, and is represented by directed graph $A = (\Pi,T)$. The set of arcs $T$ represents the tests and an arc $(i,j)$ indicates that process $i$ tests process $j$ (i.e. $j$ is the tested process in this case). Tests are executed periodically, in testing intervals determined by each tester according to its local clock. The set of process does not employ synchronized clocks nor a global clock. Thus the testing interval of a tester can be different from the testing interval of another tester. The only assumption is that local clocks move forward asymptotically. A \textit{testing round} occurs after all processes have executed their assigned tests, i.e. all tests in $T$ are performed. Testing rounds can be enumerated, $r_1, r_2, ...$, the first round corresponds to the first time tests are executed, and so on.

The purpose of diagnosis is to allow correct processes to obtain, after a finite number of testing rounds, a classification of the state of \textit{all} processes. 
Thus, $\forall i,j \in Pi$, process $j$ classifies process $i$ with function $\mbox{ \textit{state}}_j(i)$ = \{\textit{correct} $\mid$ \textit{suspect}\}.  Different processes can have different classifications about a given process, depending on the time instant in which the tests were executed, and the flow of test results among correct processes.

The latency ($L$) is defined as the number of testing rounds it takes for all correct processes to diagnose an event. Thus in em $r_e, r_{e+1}, ..., r_{e+L}$ rounds, $\forall k,j \mid \mbox{\textit{state}}(k)=$ \textit{correct} and $\mbox{\textit{state}}(j)=$ \textit{correct} $\mbox{\textit{state}}_k(i) = \mbox{\textit{state}}_j(i)$. The latency is proportional to the diameter of $A$, the testing assignment graph.

We assume that if a given process fails, then all its correct testers will suspect the process in the next testing round after the failure. This assumption is trivial to guarantee. As a crashed process does not send any reply to any test, all testers will timeout the next time the crashed process is tested. Corollary \ref{col2} makes it explicit that a test executed on a failed process always results in a suspicion.

\begin{corollary}
\label{col2}
If process $i$ has crashed and thus $state(i)$ = \textit{failed}, then for every correct process $j$ that tests $i$ $\mbox{\textit{state}}_j(i)$ = \textit{suspect}. 
\end{corollary}



The Corollary~\ref{col3} presented next determines that each process must be tested periodically by a correct process. Note that if a process is not tested, it is not possible to detect an event that might occur with that process. Thus, to ensure that all processes are properly monitored, each process must be tested by a correct process in each testing round. In each testing round $r$, $\forall i, \exists j$ such that \textit{state(j)=correct} e $(j,i) \in A$.

\begin{corollary}
\label{col3}
Each process $i \in Pi$ is tested by a correct process at each testing round, if there is one. Otherwise, it is impossible to detect an event that might occur at $i$.
\end{corollary}



In a given testing round, after a process has tested another process as correct, the tester can receive diagnostic information from the tested node. A process obtains diagnostic information by either testing other nodes or obtaining the information from correctly tested nodes.  However, if the correct process $i$ tested process $j$ in the current testing round, it should not obtain diagnostic information about $j$ from other tested processes.
Furthermore, according to Corollary~\ref{col4}, it must be ensured that each correct process receives information about all system processes in each testing round. The test assignment $A$ can be defined so that all processes can efficiently obtain all the required diagnostic information.

\begin{corollary}
\label{col4}
In each testing round, a correct process obtains diagnostic information about all other processes, either by testing those processes or by receiving information from other correctly tested processes. 
\end{corollary}

In an asynchronous system, a correct tester may suspect a correct tested process, which may be slow to respond to the test. In an extreme situation, all correct processes may suspect all other (correct) processes. In this case, all processes will test all other processes, as shown in Theorem~\ref{teo4-1}. 

\begin{theorem}
\label{teo4-1}
Assuming that each correct process executes its assigned tests once per testing round, a diagnosis algorithm specified according to the proposed model executes $n^2+n$ tests per round in the worst case. 
\end{theorem}

\begin{proof} If in a testing round all $n$ processes are correct, but everyone suspects all other $n-1$ processes, $n^2+n$ tests are performed in that round.
\end{proof}

\begin{figure}[htb]
    \centering
    \includegraphics[scale=0.65]{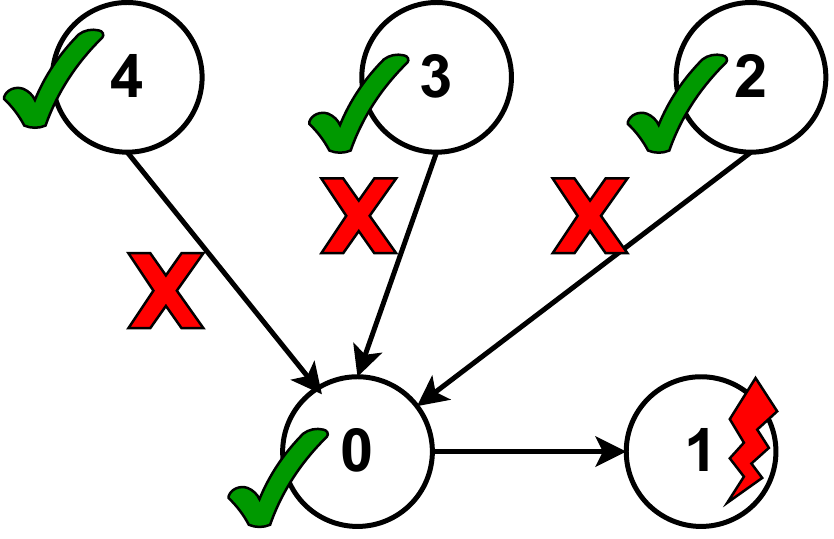}
    \caption{In this example, process 1 has failed and is tested as suspect by process 0. Eventually, correct processes 2, 3 e 4 will also suspect that process 1 has failed.}
    \label{fig2}
\end{figure}

Actually, there is no assumption on the speed at which processes execute their tests. Thus, one process can be much faster than others and execute its assigned tests multiple times in a given testing round, while the slowest process executes its assigned tests only once. In this way, the assumption in Theorem~\ref{teo4-1} that processes execute their assigned tests once per testing round guarantees that the number of tests is not greater than $n^2+n$.

In summary, the proposed model is $n-1$-diagnosable, meaning that all but one process can fail and the diagnosis can still be completed. The remaining correct process tests all others and completes the diagnosis of the system.

%% file: Sections/5-props-drct-ladc.tex
\section{Diagnosis Completeness and Accuracy}\label{sec5}

In this section, we present the two classical failure detector properties -- completeness and accuracy -- for the  proposed model. Informally, completeness requires that failed processes are suspected by correct processes. Accuracy requires that correct processes do not raise false suspicions, such that correct processes identify themselves as correct. Next, the properties are formally defined and classified as {\it weak} and {\it strong}, adapting the original definition in \cite{failure-detectors:1996}.

Diagnosis satisfies strong completeness if all correct processes classify every failed process as suspect after a finite number of testing rounds since the corresponding failure event occurred. Consider a failure event at process $i \in V$, thus i's state changes from \textit{correct} to \textit{failed}. Strong completeness requires that after a finite number of rounds $\forall j \in V$ such that $j$ is \textit{correct}, $j$ classifies $i$ as \textit{suspect}. Weak completeness, on the other hand, only requires that at least one \textit{correct} process classifies every \textit{failed} process as \textit{suspect}. Thus, for $i \in V$, after a failure event at $i$, in a finite number of rounds $\exists j \in V \mid$ \textit{state(j)=correct}, and $\mbox{\textit{state}}_j(i)=$ \textit{suspect}.

In the proposed system-level diagnosis model it is not possible to directly convert weak completeness into strong completeness as is done in \cite{failure-detectors:1996}. The reason is that the correct tester $j$ that suspected failed process $i$ can be itself \textit{suspected} (incorrectly and indefinitely) by every other correct process of the system. Figure~\ref{fig2} shows an example: process 0 which is correct tests failed process 1 and classifies it as \textit{suspect}. However, all other processes suspect 0, thus it is not possible to adopt the strategy of having the single process that correctly raised the suspicion to communicate that information to the remaining correct processes as done in \cite{failure-detectors:1996}. But this does not prevent every algorithm designed according to the proposed model to satisfy strong completeness, as shown next.

In the proposed model, in order to guarantee strong completeness it suffices to guarantee that a correct process, say $k$, either executes a test on failed process $i$ or obtain that information from another process tested correct. According to Theorem~\ref{teo1}, every algorithm specified according to the proposed model satisfies strong completeness. After a finite number of testing rounds every correct processes will have either tests the failed process suspect or obtained that information about another correct process, that in turn had either tested the failed process or obtained that information from yet another correct process, and so on.

\begin{theorem}
\label{teo1}
Any diagnosis algorithm specified according with the proposed model satisfies strong completeness.
\end{theorem}

\begin{proof}
The proof is by induction on the diameter of the testing assignment $A$. Consider a failed process $i$, such that \textit{state(i) = failed}. Induction basis: if the diameter is one, then all processes test all other processes. According to Corollary~\ref{col2}, $\forall j \in V, j \neq i$ and \textit{state(j) = correct}: 
$\mbox{\textit{state}}_j(i)=$ \textit{suspect}. Induction hypothesis: if the diameter of the testing assignment $A = d$ the completeness is strong, i.e., all processes at a distance at most of $d$ with respect to failed process $i$ effectively suspect $i$. Induction step: increment the diameter to $d+1$, including the processes that test those processes that at distance  $d$ from $i$. Those processes will obtain diagnostic information about the failure of $i$ from the tested processes.
\end{proof}

\begin{figure}[htb]
    \centering
    \includegraphics[scale=0.65]{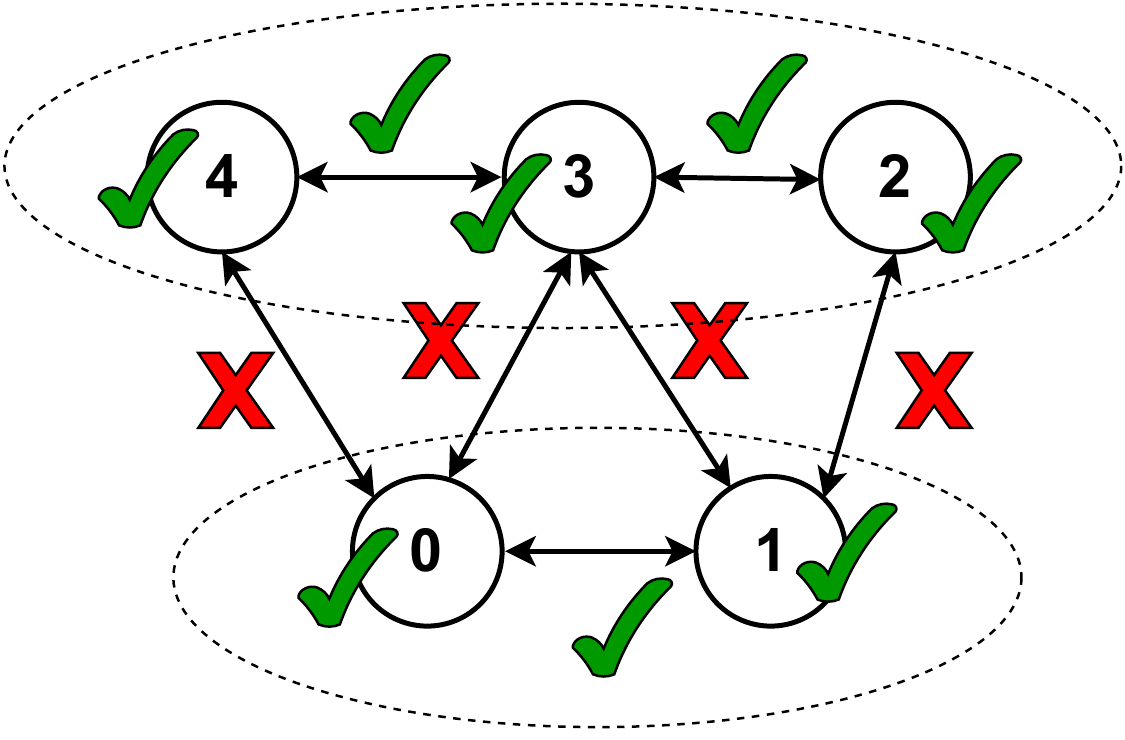}
    \caption{Precision can only be satisfied in all correct processes are strongly connected in $A$.}
    \label{fig3}
\end{figure}

Next, Theorem~\ref{teo2} proves that in order to guarantee strong accuracy in the proposed model, besides being strongly connected among themselves, correct processes cannot raise any false suspicion about each other.

\begin{theorem}
\label{teo2}
Consider an algorithm specified according with the proposed model. Strong accuracy is only guaranteed if all correct processes are strongly connected among themselves in the testing assignment $A$ and no correct process suspects any other correct process.
\end{theorem}

\begin{proof}
The proof is similar to that of Theorem~\ref{teo1}: if no correct process suspects any other correct process and all correct processes are strongly connected among themselves in $A$, then after at most $d$ testing rounds, all correct processes obtain diagnostic information about each other. In other words, $\nexists k \in V$ such that \textit{state}$(k)=$ \textit{correct} and $\mbox{\textit{state}}_k(i) =$ \textit{suspect} indefinitely.
\end{proof}

\begin{figure}[htb]
    \centering
    \includegraphics[scale=0.65]{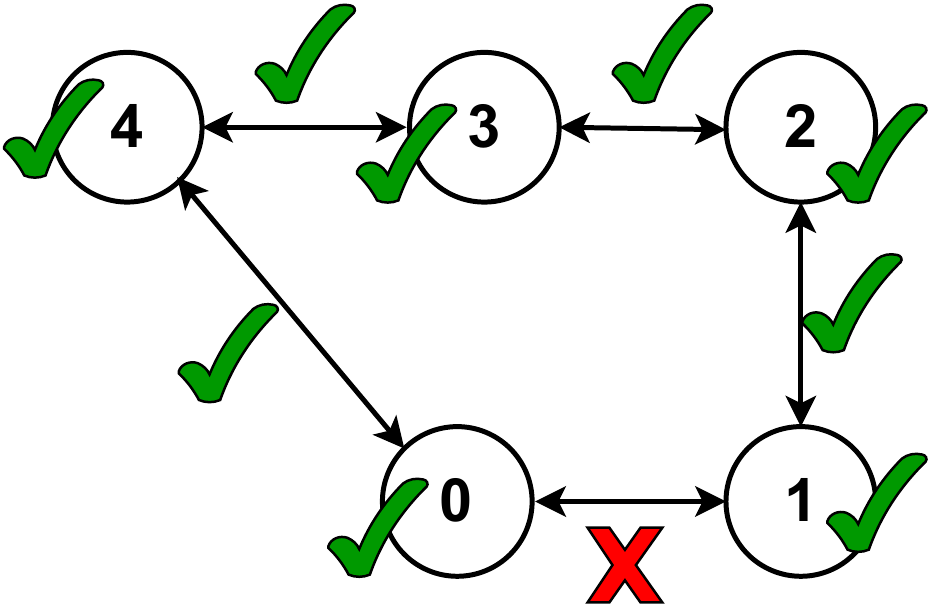}
    \caption{In this example, although all correct processes are strongly connected among themselves, processes 0 and 1 only test and suspect each other, thus the accuracy is not strong.}
    \label{fig4}
\end{figure}

It is easy to see that correct processes must be strongly connected among themselves in $A$ to guarantee strong accuracy. For example, in Figure~\ref{fig3}, processes 0 and 1 are correct, test each other, and do not suspect each other. But all other correct processes suspect both, thus all correct processes are not strongly connected, and strong accuracy is not satisfied. Theorem~\ref{teo2} shows that besides being strongly connected between themselves, no correct process can suspect any other correct process to ensure strong accuracy. For example, Figure~\ref{fig4} which shows contradictory tests for processes 0 e 1. Let $i$ be a correct process. Both $j$ and $k$ test $i$, they are the only processes that test $i$ in $A$. Suppose that $j$ determines that $i$ is correct, while $k$ suspects $i$.  Some of the other correct processes will obtain diagnostic information about $i$ from $j$ and others from $k$, so the process $i$ may remain suspect indefinitely, which breaks the accuracy.

%% file: Sections/6-latency-drct-ladc.tex
\section{Failure Detection Latency}\label{sec6}

Recall that a testing interval defines the frequency in which a process executes its assigned tests. It is an interval of time defined with the local clock. On the other hand, a testing round occurs after all correct processes have executed their assigned tests. The length of a testing round depends thus on the slowest tester. The \textit{failure detection latency} is defined as the number of testing rounds it takes to detect an event. Thus \textit{after} an event occurs on process $j$, the next testing round counts as the first for the latency, and so on until $\forall i \in \Pi$ such that $i$ is correct, $state_i(j)=$ \textit{suspect}, and the local timestamp $i$ maintains for $j$ is equal to 1.

Next, we prove theorems for the latency and failure detection for both the best- and worst- cases for each algorithm specified according to the proposed model. We assume a single event, i.e., the next event cannot occur until the previous event has been fully diagnosed, i.e., the failure has been detected by all correct processes.

\begin{theorem}
\label{teo6-1}
In the best case, the failure detection latency of any algorithm specified according to the proposed model is \textit{one} test round.
\end{theorem}

\begin{proof} After an event has occurred at process $j$, in the next testing round every correct process either tests $j$ and discovers the event, or obtains information about $j$ from another process tested correct. The order in which tests are executed in the next rounds after the event occurs affects the latency. Suppose that the first correct processes that execute tests in that round are $j$'s testers. After that, consider that all tests are executed by processes that do not test $j$ directly, but test $j$'s testers. If there are no false suspicions, these processes will obtain information about the event. Next, assume that all tests are executed on correct processes that already have information about the event. Thus in a single testing round, all processes obtain information about the event.
\end{proof}

\begin{theorem}
\label{teo6-2}
In the worst case, the failure detection latency of any algorithm specified according to the proposed model is equal to the diameter of the testing assignment $A$.
\end{theorem}

\begin{proof} 
Correct processes learn about some event either by testing failed process $j$ or by obtaining the information from another process tested correct. The proof is done by induction on the length of the testing path in $A$ through which the diagnostic information about the event propagates. Any tester $i$ of failed process $j$ detects the failure in at most one testing round. Now consider $k$, the tester of $j$'s tester $i$. Process $k$ detects the event in either 1 or 2 testing rounds. If if $k$ tests $i$ after $i$ has tested $j$, the $k$ detects the failure in the same testing round as $i$. However, if $k$ had tested $i$ before $i$ tested $k$ and detected the event, it will take an extra testing round for $k$ to make the detection.  Thus if the testing path has length 2, the latency is of at most 2 testing rounds.

Now assume that if the path has length $l$ it takes at most $l$ testing rounds for the all process along the path to detect the event. Extend the testing path by one, say with process $z$. $z$ will either detect the event in $l$ testing rounds (a test executed after the tested process had made the detection) or $l+1$ rounds, otherwise. 

\begin{figure}[htb]
    \centering
    \includegraphics[scale=0.65]{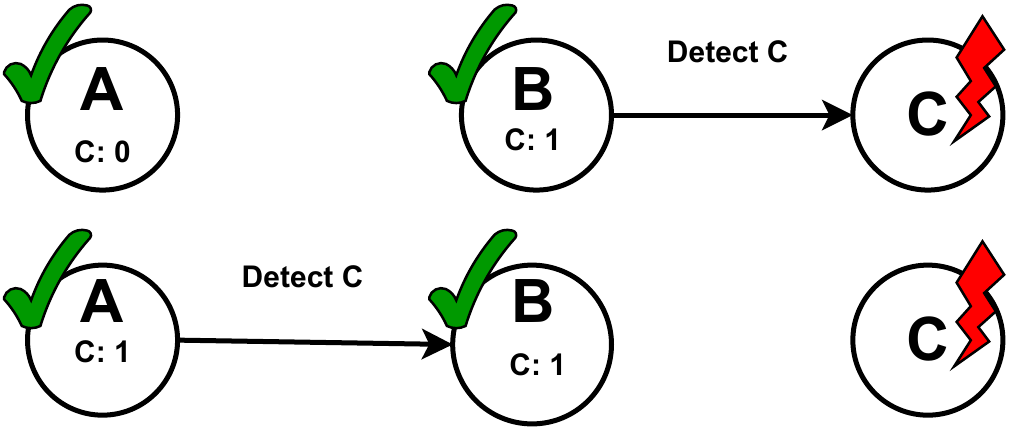}
    \caption{These tests are executed in a single testing round. As process A tests process B after it had tested process C, both detect the event in a single round.}
    \label{fig6-1}
\end{figure}

\begin{figure}[htb]
    \centering
    \includegraphics[scale=0.65]{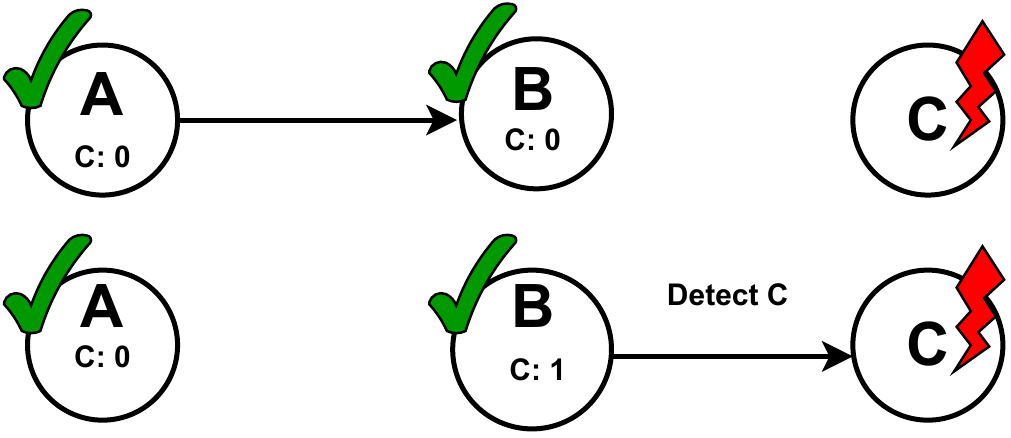}
    \caption{These tests are executed in a single testing round. As process A tests process B before process B tests process C, it will only detect the event in the next testing round.}
    \label{fig6-2}
\end{figure}

As a testing path has length at most equal to the diameter of the testing assignment graph $A$, the latency in the worst case is equal to the diameter of $A$.
\end{proof}

Figures~\ref{fig6-1} and~\ref{fig6-2} illustrate the two test situations that eventually lead to the best and worst case latencies. These tests are performed in a single testing round.  In Figure~\ref{fig6-1}, process B first tests process C and detects the event, updating the local diagnostic information accordingly. Then, process A tests process B as correct and receives diagnostic information about the event at process C. In a single testing round both processes detect the event. On the other hand, in Figure~\ref{fig6-2}, when process A tests process B, there is no information about any event, because process B will only test process C after that, detecting the event. Therefore, process A will not detect the event until the next testing round.

%% file: Sections/7-threefd-drct-ladc.tex
\section{Three Failure Detectors}\label{sec7}

In this section, we present three failure detectors specified according to the proposed system-level diagnosis model. The first is the brute-force failure detector, where each process tests every other process in every testing interval. Next, the vRing failure detector is presented, which organizes tests on a virtual ring topology. Finally, the vCube failure detector is described, which uses a hierarchical virtual topology as the testing assignment graph. 

As proved in Section~\ref{sec5}, any detector designed according to the proposed model satisfies the strong completeness condition, but can satisfy the strong accuracy condition only if no correct process suspects any other correct process, and all correct processes in the testing assignment graph $A$ are strongly connected to each other. As proved in Section~\ref{sec3}, any failure detector specified according to the proposed model requires $n^2$ tests in the worst case, but this is an extreme situation where every correct process incorrectly suspects every other correct process. All of these results are valid for the three failure detectors. The results in this section assume that there are no false suspicions. Furthermore, each correct process performs a single test per round.

We describe the three detectors and show the trade-off in terms of the number of tests required, the amount of information transferred between testers, and the latency in worst-case failure detection. In Section \ref{sec6}, it was shown that in the best case, each fault detector corresponding to the proposed model has the latency of a single testing round. The results in this section consider the worst-case latency, which is shown to be equal to the diameter of the testing assignment graph $A$.

As a process starts the execution of any of three failure detection algorithms presented, all processes are considered to be in the \textit{unknown} state, except for the tester itself that always assumes itself to be correct. The \verb+Test(i,j)+ procedure is executed by tester process $i$ on tested process $j$. \verb+Test(i,j)+ allows $i$ to classify tested process $j$ either as \textit{correct} or \textit{suspect}. A test consists of a procedure that is tailored for the system technology, and may involve a sequence of message exchanges including retransmissions, as well as the execution of multiple tasks by $j$. The \verb+Test(i,j)+ procedure encapsulates timing assumptions, as it is necessary to give up waiting for a response eventually.  As mentioned above, a correct process assumes itself to be correct, and there is no self test, i.e. $\nexists Test(i,i)$.

A tester may obtain diagnostic information from a correctly tested process. A diagnostic information consists of the process identifier and a \textit{timestamp} for the state of the tested process. All processes keep a table with the timestamps for all other processes. The timestamp is initialized as -1 (which corresponds to the state \textit{unknown}). When the process is first tested, the state is set to 0 if the tested process is \textit{correct}, otherwise it is set to $1$ (\textit{suspected}). After that, each time the tester detects an event (i.e., the tested process is suspected), the timestamp is incremented. In this paper, we assume a crash fault model with no recovery. Thus, if there are no false suspicions, the timestamp can only be -1, 0, or 1. A timestamp above 1 indicates that the corresponding process was incorrectly considered to have crashed. An even timestamp indicates that the process is considered correct, while an odd timestamp indicates that the process is suspected. The tested process does not send any information with a timestamp equal to -1, which indicates unknown.

\subsection{The Brute-Force Failure Detector}

The Brute-Force Failure Detector requires every process to directly monitor all other processes. In the proposed model, this means that each process tests all other $n-1$ processes at every testing round. Thus the \verb+Test(i,j)+ procedure is executed by every correct process $i$ at the beginning of each new testing interval on each process $j$, $\forall j \in \Pi | j \neq i$.

Next, the Brute-Force failure detector is presented in pseudo-code as executed by any process $i$.

\footnotesize
\begin{verbatim}
Algorithm Brute-Force FD executed by each correct process i
 Upon the start of a new testing interval:
   for every process j in Pi, such that j != i do
     state_j <- Test(i, j);
   Sleep until the next testing interval;
End Brute-ForceFD.
\end{verbatim}
\normalsize

The Brute-Force Failure Detector algorithm requires the execution of $n^2-n$ tests per test interval, when no process executes more than one test per round. On the other hand, the tester does not need to obtain any diagnostic information from the tested processes. The latency for detecting an event is at most 1 test round after the event has occurred. 

\subsection{The vRing Failure Detector}

The vRing (Virtual Ring) Failure Detector is inspired on the Adaptive Distributed System-level Diagnosis (Adatptive-DSD) algorithm\cite{bianchini1992implementation}. According to vRing, each process $i$ performs a test on process $i+1$. If this process tests correctly, then process $i$ receives diagnostic information about all processes except itself and the processes it tested in the current interval. Otherwise, if process $i+1$ is suspected, process $i$ tests $i+2$, $i+3$, ..., until a correct process is found or all processes are suspected. After testing a correct process and receiving diagnostic information, the tester is done for the test interval. Figure \ref{fig6-1} shows a vRing with n=6 processes, none of which have crashed and no test has raised a false suspicion. Figure \ref{fig6-2} shows the same example system, but after processes 1, 2, and 5 have crashed. So, process 0 tests 1 suspected, 2 suspected, and finally stops testing process 3 as correct, from which it receives information about processes 4 and 5. Process 4 tests process 5 as suspected, tests process 0 as correct, and receives information about processes 1, 2, and 3.

\begin{figure}[htb]
    \centering
    \includegraphics[scale=0.6]{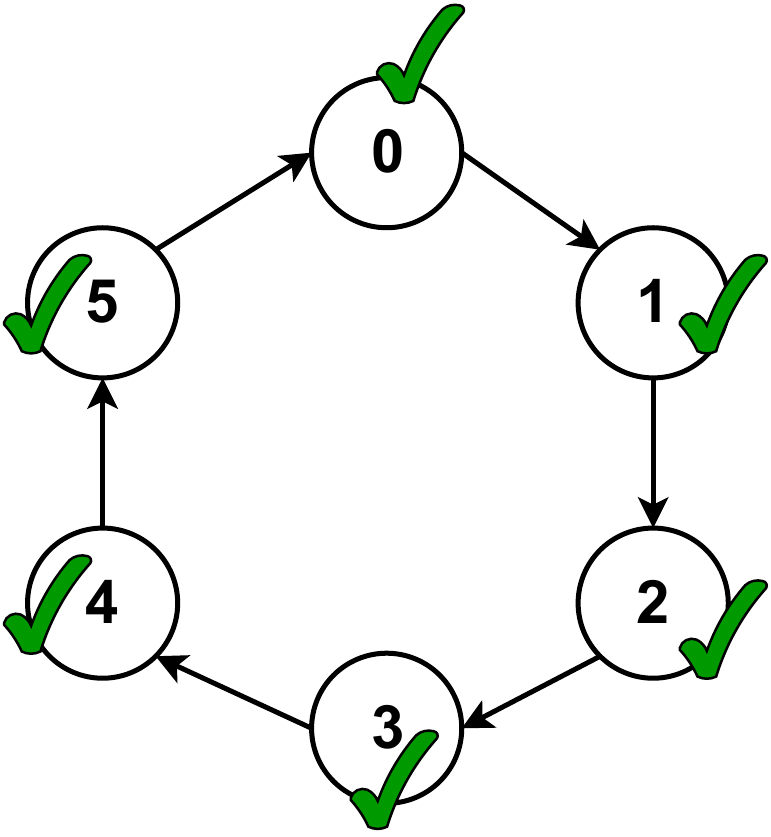}
    \caption{A vRing with n=6 processes, all correct.}
    \label{fig7-1}
\end{figure}

Next, the vRing failure detector is presented in pseudo-code as executed by any process $i$.

\footnotesize
\begin{verbatim}
Algorithm vRing FD executed by each correct process i
 Upon the start of a new testing interval:
   repeat 
     j = (i+1) mod n;
     state_j <- Test(i, j);
   until j is tested correct or every process is suspected;
   if j is correct
   then obtain new diagnostic information from j about all processes except i
   and those tested in the current interval;
   Sleep until the next testing interval;
End vRing.
\end{verbatim}
\normalsize

In the proposed model, in the best case each process running vRing tests a single process per testing round. If some processes have crashed, more tests are executed. On the other hand, if there are no false suspicions, even in the worst case each process (correct or crashed) is tested once per round. As shown in Corollary \ref{col7-1}, vRing employs the optimal (minimum) number of tests required by the proposed model, $n$ tests per testing interval. However, the worst case remains that proven in Theorem \ref{teo4-1}, every correct process $i$ tests each process $j$, $\forall j \in \Pi | j \neq i$. If all processes raise false suspicions about all other processes, $n^2-n$ tests are executed.

\begin{corollary}
\label{col7-1}
According to Corollary \ref{col3}, each process has to be tested once per testing round by a correct tester, if there is one. Thus, if there are $n$ processes, the minimum number of tests that can be executed per testing round is $n$, for if less than $n$ tests are executed one or more processes were left untested. If there is a single correct process and all others have crashed, then $n-1$ tests are executed.
\end{corollary}

\begin{figure}[htb]
    \centering
    \includegraphics[scale=0.6]{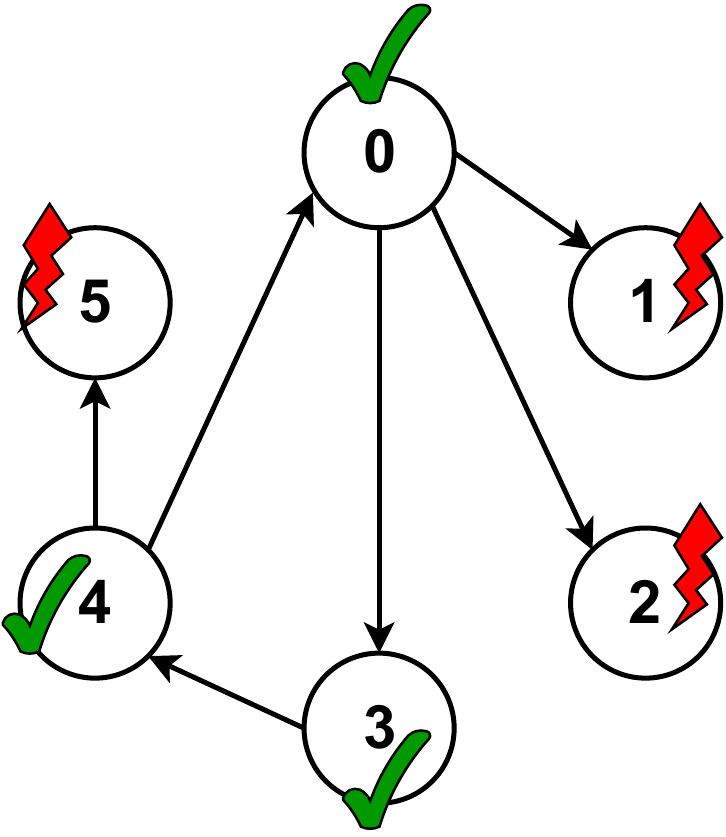}
    \caption{A vRing with n=6 processes, several have crashed.}
    \label{fig7-2}
\end{figure}

In case there are there are at least two correct processes, and no false suspicions, the vRing failure detector requires the execution of $n$ tests per testing interval. On the other hand, if all processes are correct, each tester needs to obtain diagnostic information about $n-2$ other processes from each tested processes. Thus at most $n^2-2n$ items of diagnostic information could theoretically be transferred between processes at each testing round. This amount of information can be easily reduced by having each tested processes record locally which information it had sent previously to its tester, so that the next time it is tested only new information is transferred. Furthermore, if the corresponding timestamp is equal to -1 (unknown) the tester does not get the item. Thus information is transferred just as processes learn the states of each other.

In terms of latency, according to Theorem~\ref{teo4-1}, the best case requires a single testing round, and the worst case latency is the diameter of the test assignment graph $A$,  which in this case is $n-1$ testing rounds.

\subsection{The vCube Failure Detector}

The vCube (virtual Cube) Failure Detector is inspired on the VCube distributed diagnosis algorithm \cite{duarte2014vcube}. Similar to the vRing failure detector, each process has a sequential identifier that is used as a parameter to define the tests to be executed. However, instead of a ring, vCube builds a hierarchical virtual topology that is a hypercube when all processes are correct. Figure~\ref{fig7-1} shows a vCube with $n = 8$ correct processes and no false suspicions. Bidirectional edges indicate that both processes test each other. The virtual topology corresponds to the graph $A=(V,T)$. In each testing round, each correct process $i$ performs its assigned tests, i.e., process $i$ tests process $j$ if $(i,j) \in T$. If all processes are correct and there are no suspicions, the tests correspond to the hypercube edges. Thus, each process performs $\mbox{log}n$ tests per round, for a total of $n\mbox{log}n$ tests (all logarithms are base 2).

\begin{figure}[htb]
    \centering
    \includegraphics[scale=0.60]{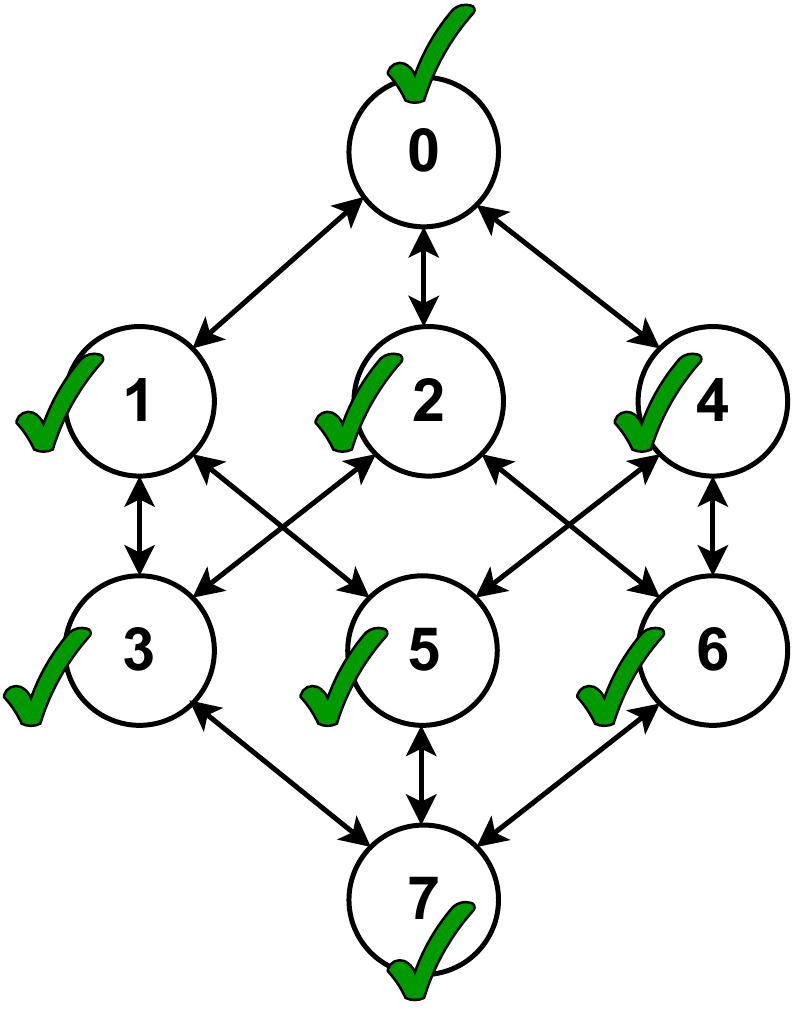}
    \caption{A vCube with no faulty processes is a hypercube.}
    \label{fig7-3}
\end{figure}

As processes crash, vCube autonomously reconfigures itself, maintaining several logarithmic properties being thus scalable by definition. vCube organizes processes in increasingly larger clusters for the purpose of testing. Each correct process $i$ executes tests on $\mbox{log}n$ clusters, which are ordered lists of processes returned by function $c(i,s) = \{ i \oplus 2^{s-1},\ _{c(i \oplus 2^{s-1}},\ 1) , ...\ ,\ _{c(i \oplus 2^{s-1}}, s-1) \}$. Table~\ref{tabelaVCube} shows all clusters for a vCube with 8 processes.

\begin{table}
	\centering
	\footnotesize
	\caption{Function $c(i,s)$ returns the clusters for $n=8$ processes.}
	\label{tabelaVCube}
	\begin{tabular}{ccccccccc}
		\toprule
		\textbf{s} & c(0,s) & c(1,s) & c(2,s) & c(3,s) & c(4,s) & c(5,s) & c(6,s) & c(7,s) \\ \midrule
		\textbf{1} & 1    & 0    & 3    & 2    & 5    & 4    & 7    & 6    \\ 
		\textbf{2} & 2,3 & 3,2 & 0,1 & 1,0 & 6,7 & 7,6 & 4,5 & 5,4 \\ 
		\textbf{3} & 4,5,6,7 & 5,4,7,6 & 6,7,4,5 & 7,6,5,4 & 0,1,2,3 & 1,0,32 & 2,3,0,1 & 3,2,1,0 \\ \bottomrule
	\end{tabular}

\end{table}

The procedure adopted to determine the tests to be executed is as follows. For each node $i$, there is an edge $(j,i)$, if $j$ is the first correct process in $c_{i,s}$, $s = 1 ... \log_2 n $. Thus each process is tested by exactly $\mbox{log}n$ testers, unless all processes of a cluster have crashed. After detecting a new event, the set of tests is recomputed. For instance, in the example shown in Figure \ref{fig7-4}, process $4$ has crashed. The tests that differ from those executed when all processes are correct are highlighted. The three tests executed by process 4 on 0, 5, and 6 are gone, but they still execute process 4. Furthermore, two new tests are added: process 5 tests processes 0 and 6, as it is the first correct process in $c(0,3)$ and $c(6,2)$. 

After testing a correct process, the tester can receive diagnostic information. The first time a process tests a process as correct, it obtains all the information the tested process has. The tester updates its local diagnostic information according with the timestamps. Diagnostic information is timestamped to allow processes to distinguish recent events, i.e., whenever the timestamp obtained for a given process is greater than the local timestamp for the same process, it is updated with the greater value. From the next time the same test is executed, the
tester only obtains information about new events.  Next, the vCube failure detector is presented in pseudo-code as executed by any process $i$.

\footnotesize
\begin{verbatim}
Algorithm vCube FD executed by each correct process i
 Upon the start of a new testing interval:
   For each assigned test (i,j) in T do
     state_j <- Test(i, j);
   If j is correct
   then obtain new diagnostic information from j about all processes except i
   and those tested in the current interval;
   If any new event was detected
   then recompute the set of assigned tests
   Sleep until the next testing interval;
End vCube.
\end{verbatim}
\normalsize

Note that this algorithm is different from the original hierarchical distributed diagnosis algorithms \cite{duarte2014vcube, duarte1998hierarchical}, since in this version each process tests all clusters in each test interval. In the distributed diagnosis algorithms, each process tests a single cluster per interval. According to Corollary~\ref{col7-2}, the maximum number of tests performed is $n\mbox{log}n$ when there are no false suspicions, since each node is tested by the first correct process of each of its $\mbox{log}n$ clusters. The number of tests will be less than $n\mbox{log}n$ only if all processes of a given cluster have crashed.  

\begin{corollary}
\label{col7-2}
Each process $i$ running vCube checks $\forall j \in V$ in which clusters $c(j,s), s=1...\mbox{log}n$ it is the first correct process. As each process $j$ has $\mbox{log}n$ clusters, if there are no false suspicions, there are at most $\mbox{log}n$ testers. Thus considering all $n$ processes, at most $n\mbox{log}n$ tests are executed per testing round. 
\end{corollary}

Regarding the latency, the best case is one testing round, as proven in Theorem~\ref{teo4-1}. Theorem~\ref{teo7-1} proves that in the worst case is of at most $\mbox{log}n$ testing rounds, even if there are false suspicions.

\begin{theorem}
\label{teo7-1}
In the worst case, the failure detection latency of vCube is of $\mbox{log}n$ testing rounds.
\end{theorem}

\begin{proof} If there are no false suspicions, the diameter of the testing assignment graph created by vCube is equal to $\mbox{log}n$ \cite{duarte2014vcube}. In case there are false suspicions, the diameter can only reduce as more edges are added to $T$.
\end{proof}

Finally, the amount of diagnostic information transferred between a tester and tested process is at most $(n-2)$ items. The worst case corresponds to a situation in which the tester obtains information about every other process (except the tester and tested processes) in a given testing round. As a tester only gets diagnostic information that is novel, and if the timestamp is not equal to -1 (unknown), the average amount of information transferred should be much less than that.

\begin{figure}
    \centering
    \includegraphics[scale=0.60]{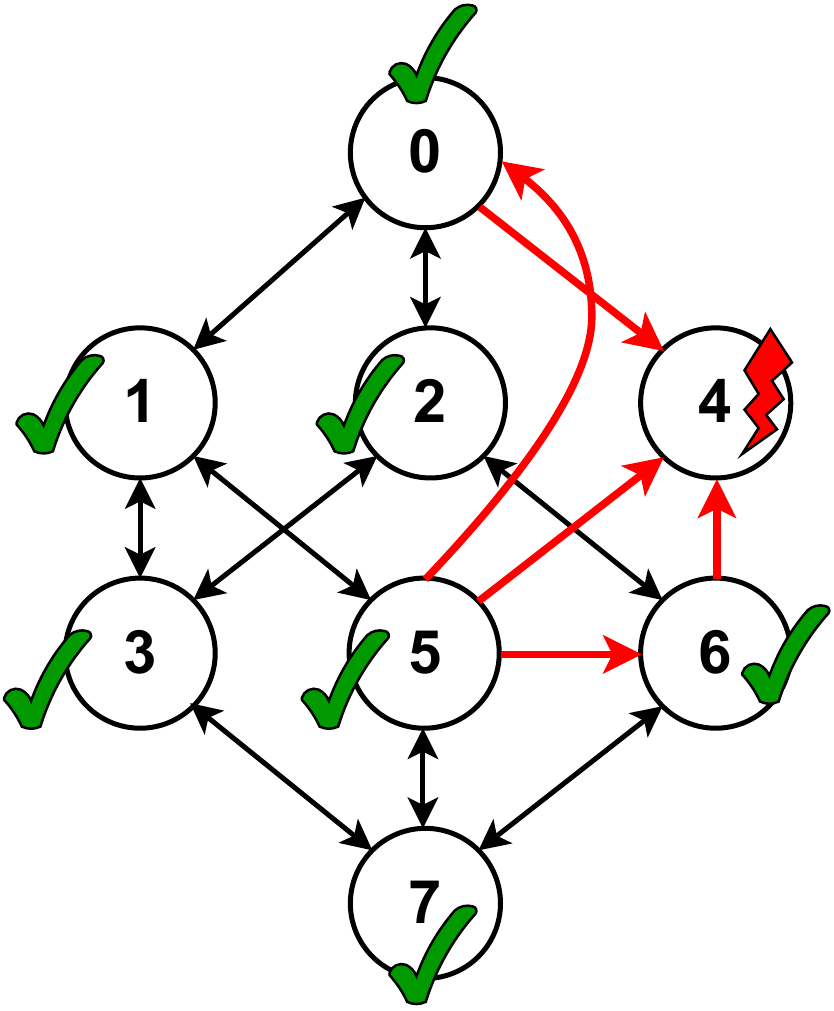}
    \caption{A vCube with n=8 processes; process 4 is crashed.}
    \label{fig7-4}
\end{figure}

\subsection{A Comparison}

We wrap up this section by presenting a brief comparison of the three failure detectors, shown in Table~\ref{tabelaComparison}. The detectors are compared based on the worst-case latency (best case is one testing round for all); number of tests per testing round in case there are no false suspicions (in the worst case all employ $n^2-n$ tests); plus the amount of information transferred per test per round in the worst case.

\begin{table}[!ht]
	\centering
	\footnotesize
	\caption{A comparison of the three failure detectors.}
	\label{tabelaComparison}
	\begin{tabular}{cccc}
		\toprule
		Algorithm & Latency (in rounds) & No. Tests & No. Data Items Transf. \\ \midrule
		Brute Force & 1 & $n^2-n$ & 0     \\
		vRing & $n-1$ & $n$ & $n-2$ \\
		vCube & $\mbox{log}n$ & $n\mbox{log}n$ & $n\mbox{log}n$ \\ \bottomrule
	\end{tabular}
\end{table}
\normalsize

In terms of the number of tests, the Brute-Force algorithm always requires $n^2-n$ tests per round, while vRing is optimal and requires $n$ tests and vCube requires $n\mbox{log}n$ tests. On the other hand, the latency of the Brute-Force algorithm is optimal as it requires one testing round for all processes to detect an event, while vRing may require up to $n-1$ testing rounds, and vCube is logarithmic. Brute Force is also the best solution in terms of the amount of information transferred between tester and tested processes, as it does not require any information. Both vRing and vCube can require up to $n-2$ items of information transferred per test, but can easily be configured to transfer only new information, which should be very small in steady state.

%% file: Sections/8-conc-drct-ladc.tex
\section{Conclusion}

Both distributed diagnosis and failure detection represent abstractions for monitoring processes in distributed systems, with both providing as output a list of processes that are considered to have crashed. The purpose of the model proposed in this work is to bridge the gap between the two abstractions by providing a common framework for defining distributed diagnosis algorithms as unreliable failure detectors. In terms of diagnosis, the biggest leap was to allow tests to give false results, i.e., a correct process can be suspected of having crashed. In addition, the classic properties of failure detectors -- completeness and accuracy -- were defined for the model and results were derived. From a failure detector perspective, several results were derived for pull-based detectors, including their performance in terms of latency, number of tests, and amount of information transmitted. Three failure detectors were presented and compared. In addition to the traditional all-monitor-all, vRing and vCube represent solutions that require the optimal number of tests (vRing) and is scalable by definition (vCube). 

Future work includes the investigation of other efficient failure detectors based on the proposed model. The model itself can be extended in several ways, for example, allowing the recovery of processes and network partitions. Furthermore, the implementation of the proposed detectors in large scale asynchronous systems should provide additional insight into the benefits of using scalable solutions for pull-based failure detection.